\documentclass[graybox]{svmult}

\usepackage{mathptmx}       % selects Times Roman as basic font
\usepackage{helvet}         % selects Helvetica as sans-serif font
\usepackage{courier}        % selects Courier as typewriter font
\usepackage{type1cm}        % activate if the above 3 fonts are
\usepackage{makeidx}         % allows index generation
\usepackage{graphicx}        % standard LaTeX graphics tool
\usepackage{multicol}        % used for the two-column index
\usepackage[bottom]{footmisc}% places footnotes at page bottom

\usepackage{amsmath,amssymb,dsfont}

\def\leqslant{\le}
\def\bq{\begin{eqnarray}}
\def\eq{\end{eqnarray}}
\def\bqq{\begin{align*}}
\def\eqq{\end{align*}}

\def\nn{\nonumber}

\def\eps{\varepsilon}

\def\thesection{\arabic{section}}
\renewcommand{\epsilon}{\varepsilon}
\newcommand\1{{\ensuremath {\mathds 1} }}

\def\Tr{{\rm Tr}}

\def\cF {\mathcal{F}}

\def\cN{\mathcal{N}}

\def\R {\mathbb{R}}

\def\I {\mathbb{I}}

\def\E {\mathcal{E}}

\def\H{\gH}

\def\R {\mathbb{R}}

\def\E {\mathcal{E}}

\def\d{\,{\rm d}}
\newcommand{\gH}{\mathfrak{H}}

\newcommand{\bH}{\mathbb{H}}
\newcommand{\dGamma}{{\ensuremath{\rm d}\Gamma}}

\allowdisplaybreaks

\makeindex             % used for the subject index
\begin{document}

\title*{Norm approximation for many-body quantum dynamics and Bogoliubov theory}
\titlerunning{Quantum dynamics and Bogoliubov theory}
\author{Phan Th\`anh Nam and Marcin Napi\'orkowski}
\authorrunning{P.T. Nam and M. Napi\'orkowski}
\institute{Phan Th\`anh Nam \at Department of Mathematics and Statistics, Masaryk University,
Kotl\'a\v rsk\'a 2, 611 37 Brno, Czech Republic, \email{ptnam@math.muni.cz}
\and Marcin Napi\'orkowski \at Department of Mathematical Methods in Physics, Faculty of Physics, University of Warsaw,  Pasteura 5, 02-093 Warszawa, Poland, \email{marcin.napiorkowski@fuw.edu.pl}
\and
{A contribution for the volume "Advances in Quantum Mechanics: contemporary trends and open problems" of the INdAM-Springer series}}

%\author[P.T. Nam]{Phan Th\`anh Nam}
%\address{Department of Mathematics and Statistics, Masaryk University,
%Kotl\'a\v rsk\'a 267/2, 611 37 Brno, Czech Republic} 
%\email{ptnam@math.muni.cz}
%\author[M. Napi\'orkowski]{Marcin Napi\'orkowski}
%\address{Department of Mathematical Methods in Physics,\\ Faculty of Physics, University of Warsaw,  Pasteura 5, 02-093 Warszawa, Poland} 
%\email{marcin.napiorkowski@fuw.edu.pl}

\maketitle

\abstract{We review some recent results on the norm approximation to the Schr\"odinger dynamics. We consider $N$ bosons in $\R^3$ with an interaction potential of the form $N^{3\beta-1}w(N^{\beta}(x-y))$ with $0\le \beta<1/2$, and show that in the large $N$ limit, the fluctuations around the condensate can be effectively described using Bogoliubov approximation.}

%\abstract{Each chapter should be preceded by an abstract (10--15 lines long) that summarizes the content. The abstract will appear \textit{online} at \url{www.SpringerLink.com} and be available with unrestricted access. This allows unregistered users to read the abstract as a teaser for the complete chapter. As a general rule the abstracts will not appear in the printed version of your book unless it is the style of your particular book or that of the series to which your book belongs.
%Please use the 'starred' version of the new Springer \texttt{abstract} command for typesetting the text of the online abstracts (cf. source file of this chapter template \texttt{abstract}) and include them with the source files of your manuscript. Use the plain \texttt{abstract} command if the abstract is also to appear in the printed version of the book.}

\section{Introduction}

In 1924-25, Bose \cite{Bose-24} and Einstein \cite{Einstein-25} predicted that at a very low temperature, many bosons condense into a common quantum state. It took 70 years until this phenomenon was first observed by Cornell, Wieman and Ketterle \cite{CorWie-95,Ketterle-95}. Since then, many interesting questions remain unsolved from the theoretical point of view. In fact, Bose and Einstein only considered the ideal gas. The study of interacting Bose gas was initiated in 1947 by Bogoliubov \cite{Bogoliubov-47}. Roughly speaking, Bogoliubov theory is based on the reduction to quasi-free states, which can be seen as the bosonic analogue to the Bardeen--Cooper--Schrieffer theory \cite{BCS-57} for superconductivity.

In the last decades, there have been many attempts to justify Bogoliubov theory from the first principles of quantum mechanics, namely from Schr\"odinger equation. In the context of the ground state problem, this has been done successfully for one and two-component Bose gases~\cite{LieSol-01,LieSol-04,Solovej-06}, for the Lee-Huang-Yang formula of homogeneous, dilute gases~\cite{ErdSchYau-08,GiuSei-09,YauYin-09} and for the excitation spectrum in the mean-field regime \cite{Seiringer-11,GreSei-13,LewNamSerSol-15,DerNap-13,NamSei-15}. In the context of the dynamical problem, Bogoliubov theory has been used widely to study the quantum dynamics of coherent states in Fock space \cite{Hepp-74,GinVel-79,GinVel-79b,RodSch-09,GriMacMar-10,GriMacMar-11,GriMac-13,Kuz-15b,BocCenSch-15,GriMac-15}. Very recently, Lewin, Schlein and one of us \cite{LewNamSch-15} were able to justify Bogoliubov theory as a norm approximation for the $N$-particle quantum dynamics in the mean-field regime. In \cite{NamNap-15,NamNap-16}, we revisited the approach in \cite{LewNamSch-15} and extended it to the case of a dilute gas. In the following, we will review our results in \cite{NamNap-15,NamNap-16} and explain the ideas of the proof.

%\subsection{Model}
We consider a system of $N$ bosons in $\R^3$, described by a wave function $\Psi_N(t)$ in the Hilbert space $\H^N=\bigotimes_{\text{sym}}^N L^2(\R^3)$. The system is governed by Schr\"odinger equation $\Psi_N(t) = e^{-itH_N}\Psi_{N}(0)$ with a typical  $N$-body Hamiltonian 
\begin{equation*} %\label{eq:def-HN}
H_N= \sum\limits_{j = 1}^N -\Delta_{x_j} + \frac{1}{N-1} \sum\limits_{1 \leqslant j < k \leqslant N} {w_N(x_j-x_k)}.
\end{equation*}
We are interested in the delta-type interaction
\begin{equation*}  % \label{eq:ass-wN}
w_N(x-y)= N^{3\beta} w(N^\beta (x-y))
\end{equation*}
where $w\ge 0$ is a fixed function which is nice enough (smooth, compact support, radially symmetric and decreasing). We put the coupling constant $1/(N-1)$ in order to make the kinetic energy and interaction energy comparable in the large $N$ limit.

The parameter $\beta \ge 0$ describes the character of the interaction between the particles. In the mean-field regime $\beta<1/3$, there are many but weak collisions and it is naturally to treat the particles as if they were independent but subjected to a common self-consistent mean-field potential. In the dilute regime $\beta>1/3$, there are few but strong collisions and the particles are more correlated. The latter regime is more relevant physically, but also more difficult mathematically. 

Our motivation is that $\Psi_N(0)$ is the ground state of a trapped system and the time evolution $\Psi_N(t)$ is observed when the trapping potential is turned off. From the rigorous result on the ground state  in \cite{LewNamSerSol-15}, we will assume that 
\begin{equation} \label{eq:PsiN0-intro}
\Psi_N(0) = \sum_{n=0}^N u(0)^{\otimes (N-n)} \otimes_s \varphi_n(0)=\sum_{n=0}^N \frac{(a^*(u(0)))^{N-n}}{\sqrt{(N-n)!}} \varphi_n(0)
\end{equation}
where $u(0)$ is a normalized function in $L^2(\R^3)$ which describes the condensate and $\Phi(0)=(\varphi_n(0))_{n=0}^\infty$ is a state in the Fock space $\cF(\{u_0\}^\bot)$ (see \eqref{eq:F+} below) for excited particles. Here we use the usual notations of the annihilation and creation operators
$$
a^*(f)=\int_{\R^3}  f(x) a_x^* \d x, \quad a(f)=\int_{\R^3} \overline{f(x)} a_x \d x, \quad \forall f\in \gH,
$$
which satisfy $[a^*_x,a^*_y]=[a_x,a_y]=0$, $[a_x,a^*_y]=\delta(x-y)$. 

When $\beta=0$, it was shown in \cite{LewNamSch-15}  that if \eqref{eq:PsiN0-intro} holds then
\begin{equation} \label{eq:PsiNt-intro}
\lim_{N\to \infty} \left\| \Psi_N(t) - \sum_{n=0}^N u(t)^{\otimes (N-n)} \otimes_s \varphi_n(t) \right\|=0
\end{equation}
(see also the recent work \cite{MitPetPic-16} for another approach). Here $u(t)$ is the evolution of the condensate, governed by Hartree equation 
\begin{equation} \label{eq:Hartree-equation}
i\partial_t u(t) =  \big(-\Delta +w_N*|u(t)|^2 -\mu_N(t)\big) u(t), \quad u(t=0)=u(0)
\end{equation}
with the phase parameter $\mu_N(t)$ which can be chosen as
$$
\mu_N(t)=\frac12\iint_{\R^3\times\R^3}|u(t,x)|^2w_N(x-y)|u(t,y)|^2 \d x \d y.
$$
The vector $\Phi(t)=(\varphi_n(t))_{n=0}^\infty$ in \eqref{eq:PsiNt-intro} is a state in the excited Fock space 
\bq \label{eq:F+}
\cF_+(t)= \cF(\{u(t)\}^\bot)= \bigoplus_{n=0}^\infty \gH_+(t)^n, \quad \gH_+(t)^n = \bigotimes^n_{\rm sym} \{u(t)\}^\bot
\eq
and its evolution is determined by Bogoliubov equation
\begin{equation} \label{eq:Bogoliubov-equation}
i\partial_t \Phi(t) = \bH(t) \Phi(t), \quad  \Phi(t=0) = \Phi(0).
\end{equation}
Here $\bH(t)$ is a quadraric Hamiltonian in Fock space: 
$$
\bH(t)= \dGamma(h(t)) + \frac12\iint_{\R^3\times\R^3}\Big(K_2(t,x,y)a^*_x a^*_y +\overline{K_2(t,x,y)}a_x a_y\Big)\d x\,\d y,
$$
which is obtained from Bogoliubov approximation (which we will explain in Section \ref{sec:prop-eff-eq}).
We use the notations $\dGamma(A)=\int a_x^* A_x a_x \d x$ (for example, $\dGamma(1)=\cN$ is the number operator) and
\begin{align*}
h(t)&=-\Delta+|u(t,\cdot)|^2\ast w_N -\mu_N(t) + Q(t) \widetilde{K}_1(t) Q(t), \\
K_2(t) &=Q(t)\otimes Q(t)\widetilde{K}_2(t), \quad Q(t)=1-|u(t) \rangle \langle u(t)|,
\end{align*}
where $\widetilde{K}_2(t,x,y)=u(t,x)w_N(x-y)u(t,y)$ is a function in $\gH^2$ and $\widetilde{K}_1(t)$ is an operator on $\gH$ with kernel $\widetilde{K}_1(t,x,y)=u(t,x)w_N(x-y)\overline{u(t,y)}$.

In order to extend \eqref{eq:PsiNt-intro} to the case $\beta>0$, we have to restrict the initial state $\Phi(0)$ in \eqref{eq:PsiN0-intro} to quasi-free states (namely the states satisfying Wick theorem) with finite kinetic energy. This reduction is again motivated by the rigorous properties of ground states in \cite{LewNamSerSol-15}. Our main result in \cite{NamNap-16} is

\begin{theorem}[Validity of Bogoliubov theory as a norm approximation] \label{thm:main}  Let $\Psi_N(t) = e^{-itH_N}\Psi_{N}(0)$ with $\Psi_N(0)$ given in \eqref{eq:PsiN0-intro}. We assume 
\begin{itemize}
\item $u(t)$ satisfies Hartree equation \eqref{eq:Hartree-equation} with the (possibly $N$-dependent) initial state $u(0,\cdot)$ satisfying $\| u(0,\cdot)\|_{W^{\ell,1}(\R^3)} \le C$ for $\ell$ sufficiently large;

\smallskip

\item $\Phi(t)=(\varphi_n(t))_{n=0}^\infty \in \cF_+(t)$ satisfies Bogoliubov equation \eqref{eq:Bogoliubov-equation} (or equivalently, equation \eqref{eq:linear-Bog-dm} in Section \ref{sec:prop-eff-eq}) with the (possibly $N$-dependent) initial state $\Phi(0)$ being a quasi-free state in $\cF_+(0)$ such that for all $\eps>0$,
\bq \label{eq:assumption-Phi0}
\big\langle \Phi(0), \cN  \Phi(0) \big\rangle \le C_\eps N^{\eps}, \quad \big\langle \Phi(0), \dGamma(1-\Delta) \Phi(0) \big\rangle\le C_\eps N^{\beta+\eps}.
\eq 
\end{itemize}
Then for all $0\le \beta<1/2$, all $\eps>0$ and all $t>0$ we have 
\begin{align} \label{eq:thm-mainresult}
\Big\| \Psi_N(t) - \sum_{n=0}^N u(t)^{\otimes (N-n)} \otimes_s \varphi_n(t) \Big\|_{\gH^N}^2 \le C_\eps (1+t)^{1+\eps} N^{(2\beta-1+\eps)/2}.
\end{align}
\end{theorem}
{\em Convention.} We always denote by $C$ (or $C_\eps$) a general positive constant independent of $N$ and $t$ ($C_\eps$ may depend on $\eps$).

There are grand canonical analogues of \eqref{eq:PsiNt-intro} related to the fluctuations around coherent states in Fock space \cite{Hepp-74,GinVel-79, GinVel-79b,GriMacMar-10,GriMacMar-11,GriMac-13,Kuz-15b}. In particular, our Theorem \ref{thm:main} is comparable to the Fock-space result of Kuz \cite{Kuz-15b}. %, and some of our ideas can be used to simplify the proof in \cite{Kuz-15b}. 
%As pointed out by Rodnianski and Schlein \cite{RodSch-09}, when $\beta=0$ (or $\beta>0$ small) and $\Psi_N(0)=u^{\otimes N}$, a norm approximation for $\Psi_N(t)$ can be obtained from the result in Fock space setting. However, this approach does not give neither optimal error estimate nor optimal range of $\beta$ (see \cite{NamNap-15} for a further discussion).
 Thanks to a heuristic argument in \cite{Kuz-15b}, the range $0\le \beta<1/2$ is optimal for the norm approximation \eqref{eq:PsiNt-intro} to hold. %This range already covers (partially) the dilute regime $\beta>1/3$. 

When $\beta>1/2$, to achieve \eqref{eq:PsiNt-intro} we have to modify the effective equations to take two-body scattering processes into account. This has been done in the Fock space setting by Boccato, Cenatiempo and Schlein \cite{BocCenSch-15} and Grillakis and Machedon \cite{GriMac-15} (see also \cite{BacBreCheFroSig-15} for a related study). Similar results for $N$-particle dynamics are still open and we hope to be able to come back to this problem in the future.

The proof of Theorem \ref{thm:main} in \cite{NamNap-16} is built up on the previous works \cite{LewNamSch-15} and \cite{NamNap-15}. The main new ingredient is the following kinetic estimate. 

\begin{theorem}[Kinetic estimate] \label{lem:HN-kinetic} Let $\Psi_N(0)$ as in Theorem \ref{thm:main}. Then for all $0<\beta<1/2$, all $\eps>0$ and all $t>0$, we have 
\bq \label{eq:new-kinetic-leadingorder}
\big \langle \Psi_N(t), \dGamma(Q(t)(1-\Delta)Q(t))  \Psi_N(t) \big\rangle \le  C_\eps  (N^{\beta+\eps}+N^{3\beta-1}).
\eq
\end{theorem}
We can introduce the density matrix $\gamma_{\Psi_N(t)}^{(1)}:\gH\to \gH$ with kernel  $\gamma_{\Psi_N(t)}^{(1)}(x,y)= \langle \Psi_N(t), a_y^* a_x  \Psi_N(t) \rangle$ and rewrite \eqref{eq:new-kinetic-leadingorder} as 
\bq \label{eq:CS-1}
\Tr \Big( \sqrt{1-\Delta}Q(t) \gamma_{\Psi_N(t)}^{(1)}Q(t) \sqrt{1-\Delta} \Big) \le  C_\eps  (N^{\beta+\eps}+N^{3\beta-1}).
\eq
By the Cauchy-Schwarz inequality, \eqref{eq:CS-1} implies that for all $0<\beta<2/3$,
\bq \label{eq:cv-tr-gamma1}
\lim_{N\to \infty} \Tr \Big| \sqrt{1-\Delta} \Big(N^{-1}\gamma_{\Psi_N}^{(1)} -  |u(t) \rangle \langle u(t)| \Big)\sqrt{1-\Delta} \Big| =0
\eq
(see Section \ref{sec:kinetic-bound} for more details). In case $\beta=0$, the approximation of the form \eqref{eq:cv-tr-gamma1} has been studied in \cite{MicSch-12,Luhrmann-12,AnaHot-16,MitPetPic-16}.  Note that \eqref{eq:cv-tr-gamma1} is stronger than the standard definition of the Bose-Einstein condensation
\bq \label{eq:cv-tr-gamma2}
\lim_{N\to \infty} \Tr \Big| N^{-1}\gamma_{\Psi_N}^{(1)} -  |u(t) \rangle \langle u(t)| \Big| =0
\eq
which has been studied by many authors; see  \cite{Spohn-80,BarGolMau-00,ErdYau-01,AdaGolTet-07,ErdSchYau-07} for some pioneer works (in these works, the convergence \eqref{eq:cv-tr-gamma2} was derived using the BBGKY hierarchy, a method that is less quantitative than our approach). 

Note that when $\beta=1$ (the Gross--Pitaevskii regime), the strong correlations between particles require a subtle correction: the nonlinear term $w_N*|u(t)|^2$ in Hartree equation \eqref{eq:Hartree-equation} has to be replaced by $8\pi a|u(t)|^2$ with $a$ being the scattering length of $w$. This has been justified rigorously in the context of the Bose-Einstein condensation \eqref{eq:cv-tr-gamma2}; see  \cite{LieSeiYng-00,LieSei-06,NamRouSei-15} for the ground state problem and \cite{ErdSchYau-10,ErdSchYau-09,BenOliSch-15,Pickl-15} for the dynamical problem. The norm approximation is completely open.

In the rest, we discuss Hartree and Bogoliubov equations in Section \ref{sec:prop-eff-eq}, and then go to the proofs of Theorems \ref{lem:HN-kinetic} and Theorem \ref{thm:main} in Sections \ref{sec:kinetic-bound} and   \ref{sec:sketch-proof}, respectively.

\section{Effective equations}\label{sec:prop-eff-eq}

We recall the well-posedness of Hartree equation from  \cite[Prop. 3.3 \& Cor. 3.4]{GriMac-13}.

\begin{lemma}[Hartree equation] \label{lem:Hartree-equation} If $u(0,\cdot)\in H^2(\R^3)$, then  Hartree equation \eqref{eq:Hartree-equation} has a unique global solution $u \in C( [0,\infty),H^2(\R^3)) \cap C^1((0,\infty),L^2(\R^3))$. 
Moreover, if $\|u(0,\cdot)\|_{W^{\ell,1}(\R^3)}\le C$ with $\ell$ sufficiently large, then 
$$ \|u(t,\cdot)\|_{H^2} \le C,\quad \|\partial_t u(t,\cdot)\|_{L^2} \le C,\quad \|u(t, \cdot)\|_{L^\infty} + \|\partial_t u(t,\cdot)\|_{L^\infty} \le C(1+t)^{-3/2}.$$
\end{lemma}

As in \cite[Sec. 2.3]{LewNamSerSol-15},  any vector $\Psi\in \gH^N$  can be written uniquely as 
\begin{equation*}
\Psi=\sum_{n=0}^N u(t)^{\otimes (N-n)} \otimes_s \varphi_n  = \sum_{n=0}^N \frac{(a^*(u(t)))^{N-n}}{\sqrt{(N-n)!}} \varphi_n
\end{equation*}
with $\varphi_n \in \gH_+(t)^{n}$. This gives rise to the  unitary operator $U_{N}(t): \gH^N \to \1^{\le N} \cF_+(t)$
\begin{equation*}
U_N(t) \Psi = \varphi_0\oplus \varphi_1 \oplus\cdots \oplus \varphi_N.
\end{equation*}
Here $\1^{\le N}$ for the projection onto $\mathbb{C} \oplus \gH \oplus \cdots \oplus \gH^N$. Some fundamental properties of $U_N(t)$ can be found in \cite[Proposition 4.2]{LewNamSerSol-15} and \cite[Lemma 6]{LewNamSch-15}.

Next, as in \cite{LewNamSch-15}, we introduce $\Phi_N(t):=U_N(t) \Psi_N(t)$ and rewrite Schr\"odinger equation  as 
\bq \label{eq:eq-PhiNt}
i \partial_t \Phi_N(t)  =  \widetilde H_N (t)   \Phi_N(t), \quad  \Phi_N(0) = \1^{\le N} \Phi(0).
\eq
Here $
\widetilde H_N (t)=  \1^{\le N} \Big[ \bH(t) + \frac{1}{2}\sum_{j=0}^4 ( R_{j} + R_j^*) \Big] \1^{\le N}
$
with 
\begin{align*}
R_{0}&= \d\Gamma(Q(t)[w_N*|u(t)|^2+ \widetilde{K}_1(t) -\mu_N(t)]Q(t))\frac{1-\cN}{N-1},\\
R_{1}&=-2\frac{\cN\sqrt{N-\cN}}{N-1} a(Q(t)[w_N*|u(t)|^2]u(t)),\\
R_{2}&= \iint  K_2(t,x,y) a^*_x a^*_y \d x \d y \left(\frac{\sqrt{(N-\cN)(N-\cN-1)}}{N-1}-1\right),\\
R_{3}& =\frac{\sqrt{N-\cN}}{N-1}\iiiint( 1 \otimes Q(t) w_N Q(t)\otimes Q(t))(x,y;x',y')\times \\
& \qquad \qquad \qquad \qquad \qquad \qquad \qquad  \times \overline{u(t,x)} a^*_y a_{x'} a_{y'} \,\d x \d y \d x' \d y',\\
R_{4}&= \frac{1}{2(N-1)}\iiiint({Q(t)}\otimes{Q(t)}w_N Q(t)\otimes Q(t))(x,y;x',y')\times \\
&  \qquad \qquad \qquad \qquad \qquad \qquad \qquad  \times a^*_x a^*_y a_{x'} a_{y'} \,\d x \d y \d x' \d y'.
\end{align*}
(In $R_0$ and $R_1$ we write $w_N$ for the function $w_N(x)$, while in $R_3$ and $R_4$ we write $w_N$ for the two-body multiplication operator $w_N(x-y)$.) 

The idea of Bogoliubov approximation is that when $N\to \infty$ all error terms $R_j$'s are so small that we can ignore them and replace \eqref{eq:eq-PhiNt} by Bogoliubov equation \eqref{eq:Bogoliubov-equation}. Some important properties of this equation are collected in the following

\begin{lemma}[Bogoliubov equation] \label{lem:Bogoliubov-equation} (i) If $\Phi(0)$ belongs to the quadratic form domain $\mathcal{Q}(\dGamma (1-\Delta))$, then equation \eqref{eq:Bogoliubov-equation} has a unique global solution in $\mathcal{Q}(\dGamma (1-\Delta))$.
%$$\Phi \in C([0,\infty), \cF(\gH)) \cap L^\infty_{\rm loc} ((0,\infty), \mathcal{Q}(\dGamma (1-\Delta))).$$
Moreover, the pair of density matrices $(\gamma_{\Phi(t)},\alpha_{\Phi(t)})$ is the unique solution to
\bq \label{eq:linear-Bog-dm} 
\left\{
\begin{aligned}
i\partial_t \gamma &= h \gamma - \gamma h + K_2 \alpha - \alpha^* K_2^*, \\
i\partial_t \alpha &= h \alpha + \alpha h^{\rm T} + K_2  + K_2 \gamma^{\rm T} + \gamma K_2,\\
\gamma(t&=0)=\gamma_{\Phi(0)}, \quad \alpha(t=0)  = \alpha_{\Phi(0)}.
\end{aligned}
\right.
\eq
(ii) If we assume further that $\Phi(0)$ is a quasi-free state in $\cF_+(0)$, then $\Phi(t)$ is a quasi-free state in $\cF_+(t)$ for all $t>0$ and 
\bq \label{eq:Bog-N}
\langle \Phi(t), \cN \Phi(t) \rangle \le C \Big( \langle \Phi(0),\cN \Phi(0)\rangle^2 + [\log(2+t)]^2\Big).
\eq
\end{lemma}

Recall that $\gamma_{\Phi(t)}:\gH \to \gH$, $\alpha_{\Phi(t)}:\overline{\gH} \equiv \gH^* \to {\gH}$ are operators with kernels
$\gamma_{\Phi(t)}(x,y)=\langle \Phi(t), a_y^* a_x \Phi(t) \rangle$, $\alpha_{\Phi(t)}(x,y)=\langle \Phi(t), a_x a_y \Phi(t) \rangle
$ and $K_2: \overline{\gH} \equiv \gH^*\to \gH$ is an operator with kernel $K_2(t,x,y)$. Note that \eqref{eq:linear-Bog-dm} is similar (but not identical) to the effective equations used in  the Fock space setting in  \cite{GriMac-13,Kuz-15b}. 
 
\begin{proof} \smartqed (i) For existence and uniqueness of $\Phi(t)$, we refer to \cite[Theorem 7]{LewNamSch-15}. To derive \eqref{eq:linear-Bog-dm}, we use \eqref{eq:Bogoliubov-equation} to compute 
\begin{align*}
&i\partial_t \gamma_{\Phi(t)}(x',y') = i\partial_t \langle \Phi(t), a_{y'}^* a_x \Phi(t) \rangle = \langle \Phi(t), [a_{y'}^* a_x , \bH(t)] \Phi(t) \rangle \\
&= \iint h(t,x,y) \Big(\delta (x'-x) \gamma_{\Phi(t)}(y,y') - \delta (y'-y) \gamma_{\Phi(t)}(x',x) \Big) \d x  \d y \\
& + \frac{1}{2} \iint k(t,x,y) \Big( \delta(x'-x) \alpha_{\Phi(t)}^*(y,y') + \delta (x'-y) \alpha_{\Phi(t)}^*(y',x) \Big) \d x  \d y \\
& - \frac{1}{2} \iint k^*(t,x,y) \Big( \delta(y'-y) \alpha_{\Phi(t)}(x,x') + \delta (y'-x) \alpha_{\Phi(t)}(y,x') \Big) \d x  \d y \\
& = \Big( h(t) \gamma_{\Phi(t)} - \gamma_{\Phi(t)} h(t) + K_2(t)\alpha^*_{\Phi(t)} - \alpha_{\Phi(t)} K_2^*(t) \Big)(x',y')
\end{align*}
This is the first equation in  \eqref{eq:linear-Bog-dm}. The second equation is proved similarly. 

\medskip

\noindent (ii) Now we show that if $\Phi(0)$ is a quasi-free state, then $\Phi(t)$ is a quasi-free state for all $t>0$. We will write $(\gamma,\alpha)=(\gamma_{\Phi(t)}, \alpha_{\Phi(t)})$ for short. Let us introduce 
$$X:=\gamma+ \gamma^2 -\alpha \alpha^*, \quad Y:= \gamma \alpha - \alpha \gamma^{\rm T}.$$
It is a general fact (see, e.g., \cite[Lemma 8]{NamNap-15}) that $\Phi(t)$ is a quasi-free state if and only if $X(t)=0$ and $Y(t)=0$. In particular, we have $X(0)=0$ and $Y(0)=0$ by the assumption on $\Phi(0)$. Using \eqref{eq:linear-Bog-dm} it is straightforward to see that  
\begin{align*}
i\partial_t X &= h X- X h + k Y^* - Y k^* ,\\
i\partial_t X^2 &= (i\partial_t X) X + X (i\partial_t X) = hX^2 - X^2 h + (K_2Y^* - Y K_2^*)X + X(K_2Y^* - Y K_2^*).
\end{align*}
Then we take the trace and use $\Tr(h X^2 - X^2 h)=0$ ($h X^2$ and $X^2 h$ may be not trace class but we can introduce a cut-off; see \cite{NamNap-15} for details). We find that
$$
\|X(t)\|_{\rm HS}^2  \le 4 \int_0^t \|K_2(s)\|\cdot \|X(s)\|_{\rm HS}\cdot \|Y(s)\|_{\rm HS} \,  \d s
$$
We also obtain a similar bound for $\|Y(t)\|_{\rm HS}$. Then summing these estimates and using the fact that $\|K_2(t)\|$ is bounded uniformly in time, we conclude by Gr\"onwall's inequality that $X(t)=0$, $Y(t)=0$ for all $t>0$ . 

A similar argument can be used to the uniqueness of solutions to \eqref{eq:linear-Bog-dm}. 

To obtain \eqref{eq:Bog-N}, we first  estimate $\|\alpha\|_{\rm HS}^2+\|\gamma\|_{\rm HS}^2$ by a Gr\"onwall-type inequality, and then use the identity $\|\alpha\|_{\rm HS}^2=\Tr(\gamma+\gamma^2)$. We refer to \cite{NamNap-15} for details. \qed
\end{proof}

\section{Kinetic bounds} \label{sec:kinetic-bound}

In this section, we discuss  Theorem \ref{lem:HN-kinetic}. As mentioned, it is equivalent to \eqref{eq:CS-1} and in case $\beta<2/3$ it implies \eqref{eq:cv-tr-gamma1}. Let us explain the implication from \eqref{eq:CS-1} to \eqref{eq:cv-tr-gamma1} in more details. We will write $P(t)=|u(t)\rangle \langle u(t)|$ for short. We can decompose
\begin{align*}
 N^{-1}\gamma_{\Psi_N(t)}^{(1)} - P(t)  &= N^{-1}Q(t) \gamma_{\Psi_N(t)}^{(1)} Q(t)  - N^{-1}\Tr \Big(Q(t) \gamma_{\Psi_N(t)}^{(1)} Q(t)\Big) P(t) \\
 & \quad +  N^{-1}Q(t) \gamma_{\Psi_N(t)}^{(1)} P(t) +  N^{-1}P(t) \gamma_{\Psi_N(t)}^{(1)} Q(t)  
\end{align*}
and use the triangle inequality of the trace norm to estimate
\begin{align} \label{eq:decomp}
&\Tr \Big| \sqrt{1-\Delta} \Big(N^{-1}\gamma_{\Psi_N}^{(1)} -  |u(t) \rangle \langle u(t)| \Big)\sqrt{1-\Delta} \Big| \nn  \\
& \le N^{-1}\Tr \Big( \sqrt{1-\Delta}Q(t) \gamma_{\Psi_N(t)}^{(1)}Q(t) \sqrt{1-\Delta} \Big) + N^{-1}\Tr \Big(Q(t) \gamma_{\Psi_N(t)}^{(1)} Q(t)\Big) \|u(t,\cdot)\|_{H^1}^2 \nn\\
&\quad + 2 N^{-1} \Tr \Big| \sqrt{1-\Delta} Q(t) \gamma_{\Psi_N(t)}^{(1)} P(t) \sqrt{1-\Delta} \Big|  .\end{align}
Using the Cauchy-Schwarz inequality (for Schatten norm)
\begin{align*}
&\Tr  \left| (1-\Delta)^{1/2}Q(t) \gamma_{\Psi_N(t)}^{(1)} P(t) (1-\Delta)^{1/2} \right| \\
&\le \Big\| (1-\Delta)^{1/2}Q(t) \Big(\gamma_{\Psi_N(t)}^{(1)} \Big)^{1/2} \Big\|_{\rm HS} \cdot \Big\| \Big(\gamma_{\Psi_N(t)}^{(1)} \Big)^{1/2} \Big\| \cdot \Big\|P(t) (1-\Delta)^{1/2}\Big\|_{\rm HS}
\end{align*}
we deduce from \eqref{eq:CS-1} and \eqref{eq:decomp} that for all $\eps>0$,
\bq \label{eq:CS--a}
\Tr \Big| \sqrt{1-\Delta} \Big(N^{-1}\gamma_{\Psi_N}^{(1)} -  |u(t) \rangle \langle u(t)| \Big)\sqrt{1-\Delta} \Big| \le C_\eps N^{a+\eps}
\eq
where $a=\max\{ \beta-1, (\beta-1)/2, 3\beta-2, (3\beta-2)/2 \}$. If  $\beta<2/3$, then \eqref{eq:cv-tr-gamma1} holds.

Now we turn to another version of Theorem \ref{lem:HN-kinetic}. From the definition $\Phi_N(t)=U_N(t)\Psi_N(t)$, we can check  that  $
Q(t)\gamma_{\Psi_N}^{(1)}Q(t)= \gamma_{\Phi_N}^{(1)}
$ (e.g. by using \cite[Proposition 4.2]{LewNamSerSol-15}). Thus Theorem \ref{lem:HN-kinetic} is equivalent to 

\begin{theorem}[Kinetic estimate] \label{lem:wHN-kinetic} Let $\Phi_N(t)$ be as in \eqref{eq:eq-PhiNt}, with $\Phi(0)$ as in Theorem \ref{thm:main}. Then for all $\eps>0$ and all $t>0$, we have 
\bq \label{eq:new-kinetic}
\big \langle \Phi_N(t), \dGamma(1-\Delta)  \Phi_N(t) \big\rangle \le  C_\eps  (N^{\beta+\eps}+N^{3\beta-1+\eps}).
\eq
\end{theorem}

Before proving Theorem \ref{lem:wHN-kinetic}, let us start with a simpler bound.

\begin{lemma}[Bogoliubov kinetic bound] \label{lem:bH-kinetic} Let $\Phi(t)$ be as in Theorem \ref{thm:main}. Then 
$$
\big \langle \Phi(t), \dGamma(1-\Delta)  \Phi(t) \big\rangle \le  C_\eps  N^{\beta+\eps} ,\quad \forall t>0.
$$
\end{lemma}

\begin{proof} \smartqed For a general quadratic Hamiltonian, we have 
$$ \dGamma(H) + \frac{1}{2} \iint \Big( K(x,y) a_x^* a_y^* + \overline{K(x,y)}a_x a_y \Big) \d x \d y \ge -\frac{1}{2} \iint |(H_x^{-1/2} K(x,y)|^2 \d x \d y.$$
This bound can be found in our recent joint work with Solovej \cite[Lemma 9]{NamNapSol-16} (see also \cite[Theorem 5.4]{BruDer-07} for a similar result). Combining this with the Sobolev-type estimate (see \cite[Lemma 6]{NamNap-16})
\begin{align*}
\|(1-\Delta_x)^{-1/2} K_2(t, \cdot, \cdot)\|^2_{L^2} + \|(1-\Delta_x)^{-1/2} \partial_t K_2(t, \cdot, \cdot)\|^2_{L^2} \le C_\eps (1+t)^{-3} N^{\beta+\eps} 
\end{align*}
we obtain the quadratic form inequalities (see \cite[Lemma 7]{NamNap-16})
\begin{align}
\pm \Big( \bH(t) + \dGamma(\Delta) \Big) &\le \eta \dGamma(1-\Delta) +  \frac{C_\eps(\cN+N^{\beta+\eps})}{\eta(1+t)^3}, \label{eq:BogHam-bounds1}\\
\pm \partial_t \bH(t) &\le \eta \dGamma(1-\Delta) +  \frac{C_\eps (\cN  +  N^{\beta+\eps})}{ \eta (1+t)^3},\label{eq:BogHam-bounds2}\\
\pm i[\bH(t),\cN] &\le \eta \dGamma(1-\Delta) + \frac{C_\eps (\cN  +  N^{\beta+\eps})}{ \eta (1+t)^3} \label{eq:BogHam-bounds3}
\end{align}
for all $\eta>0$. On the other hand, from Bogoliubov equation \eqref{eq:Bogoliubov-equation}, we have
\begin{align} \label{eq:Phit-Grw}
\big\langle \Phi(t), \bH (t) \Phi(t)  \big\rangle - \big\langle \Phi(0), \bH (0) \Phi(0)  \big\rangle = \int_0^t \big\langle \Phi(s), \partial_s \bH (s) \Phi(s)  \big\rangle \d s. 
\end{align}
Using \eqref{eq:BogHam-bounds1} with $\eta=1/2$ we have $\big\langle \Phi(0), \bH (0) \Phi(0)  \big \rangle  \le C_\eps N^{\beta+\eps}$ and
$$
\big\langle \Phi(t), \bH (t) \Phi(t)  \big\rangle  \ge  \frac{1}{2} \big \langle \Phi(t), \dGamma(1-\Delta)  \Phi(t) \big\rangle - C_\eps \Big( \big \langle \Phi(t), \cN  \Phi(t) \big\rangle +  N^{\beta+\eps} \Big)
$$
Using \eqref{eq:BogHam-bounds2} with $\eta=(1+t)^{-3/2}$ we get
$$
\big\langle \Phi(t), \partial_t \bH (t)\Phi(t)  \big\rangle \le C_\eps (1+t)^{-3/2} \Big( \big\langle \Phi(t), \dGamma(1-\Delta) \Phi(t) \big\rangle+ N^{\beta+\eps} \Big) .
$$
Thus \eqref{eq:Phit-Grw} implies that
\begin{align} \label{eq:Phit-Grw-1}
\big \langle \Phi(t), \dGamma(1-\Delta)  \Phi(t) \big\rangle & \le  C_\eps \int_0^t (1+s)^{-3/2} \big \langle \Phi(s), \dGamma(1-\Delta)  \Phi(s) \big\rangle \d s \nn\\ &\qquad \qquad + C_\eps \Big( \big \langle \Phi(t), \cN  \Phi(t) \big\rangle + N^{\beta+\eps} \Big). 
\end{align}
Similarly, we can estimate $\partial_t \langle \Phi(t), \cN \Phi(t) \rangle$ by using Bogoliubov equation \eqref{eq:Bogoliubov-equation} and \eqref{eq:BogHam-bounds3} with $\eta=(1+t)^{-3/2}$. Then we integrate the resulting bound and obtain
\begin{align*}
\langle \Phi(t), \cN \Phi(t) \rangle \le C_\eps \int_0^t (1+s)^{-3/2}  \big\langle \Phi(s), \dGamma(1-\Delta) \Phi(s) \big\rangle \d s + C_\eps N^{\beta+\eps}.
\end{align*}
Inserting the latter inequality into the right side of \eqref{eq:Phit-Grw-1} we obtain
\begin{align*} 
\big \langle \Phi(t), \dGamma(1-\Delta)  \Phi(t) \big\rangle & \le C_\eps (1+s)^{-3/2} \int_0^t  \big \langle \Phi(s), \dGamma(1-\Delta)  \Phi(s) \big\rangle  \d s + C_\eps N^{\beta+\eps}. 
\end{align*}
The desired result then follows from a  Gronwall-type inequality.  \qed
\end{proof}

The proof of Theorem \ref{lem:wHN-kinetic} is based on a similar argument. We will use the following estimates on the error terms $R_j$'s in \eqref{eq:eq-PhiNt} (see  \cite[Lemmas 9, 11]{NamNap-16}).

\begin{lemma}[Control of error terms]\label{lem:Rj} Let $R_j$'s be as in \eqref{eq:eq-PhiNt}. Then we have the quadratic form estimates on $\1^{\le N} \cF_+(t)$: 
\begin{align*}
\pm (R_j+R_j^*) \le &\eta \Big( R_4 + \frac{\cN^2}{N} \Big)+ \frac{C (1+\cN)}{\eta(1+t)^3}, \quad \forall \eta>0, \,\forall j=0,1,2,3, \\
 0\le R_4 &\le CN^{3\beta-1} \cN^2, \quad R_4\le CN^{\beta-1} \dGamma(-\Delta) \cN,\\
 \pm \partial_t (R_j+R_j^*) &\le \eta \Big( R_4 + \frac{\cN^2}{N} \Big)+ \frac{C (1+\cN)}{\eta(1+t)^3},\, \forall j=0,1,2,3,4,\\
\pm i[(R_j+R_j^*),\cN] &\le \eta \Big( R_4 + \frac{\cN^2}{N} \Big)+ \frac{C (1+\cN)}{\eta(1+t)^3},\, \forall j=0,1,2,3,4.
\end{align*}
\end{lemma}

Now we are ready to provide

\begin{proof}[of Theorem \ref{lem:wHN-kinetic}] \smartqed From \eqref{eq:eq-PhiNt} we have
\begin{equation} \label{eq:PhiN-kin-0}
\big\langle \Phi_N(t), \widetilde H_N(t) \Phi_N(t) \big\rangle - \big\langle \Phi_N(0), \widetilde H_N(0) \Phi_N(0) \big\rangle =  \int_0^t  \! \big\langle \Phi_N(s), \partial_s \widetilde H_N(s) \Phi_N(s) \big\rangle \d s .
\end{equation}
Using \eqref{eq:BogHam-bounds1}  and Lemma \ref{lem:Rj}, we can estimate 
\begin{align*}
\big\langle \Phi_N(t), \widetilde H_N(t) \Phi_N(t) \big\rangle & \ge \frac{1}{2} \big\langle \Phi_N(t), (\dGamma(1-\Delta) + R_4) \Phi_N(t) \big\rangle  \nn\\
&\qquad\qquad- C_\eps \Big(N^{\beta+\eps} + \big\langle \Phi_N(t), \cN  \Phi_N(t) \big\rangle \Big), \\
\big\langle \Phi_N(0), \widetilde H_N(0) \Phi_N(0) \big\rangle & \le C_\eps (N^{\beta+\eps}+N^{3\beta-1+\eps}).
\end{align*}
Here in the last inequality, we have used $R_4\le CN^{3\beta-1}\cN^2$ (see Lemma \ref{lem:Rj}) and a well-known moment estimate for every quasi-free state $\Phi$:
\bq \label{eq:moment-quasi-free}
\Big\langle \Phi, (1+\cN)^{s} \Phi \Big\rangle \le C_{s} \Big\langle \Phi, (1+\cN) \Phi \Big\rangle^{s}
\eq
where the constant $C_s$ depends only on $s\in \mathbb{N}$ (see e.g. \cite[Lemma 5]{NamNap-15}). Moreover, from \eqref{eq:BogHam-bounds2} and Lemma \ref{lem:Rj} we obtain 
\begin{align*} % \label{eq:PhiN-kin-3}
\big\langle \Phi_N(t), \partial_t \widetilde H_N(t) \Phi_N(t) \big\rangle  \le C_\eps (1+t)^{-3/2} \Big( \big\langle \Phi_N(t), (\dGamma(1-\Delta) + R_4) \Phi_N(t) \big\rangle+N^{\beta+\eps} \Big).
\end{align*}
Thus \eqref{eq:PhiN-kin-0} implies that 
\begin{align} \label{eq:PhiN-kin-4}
\big\langle \Phi_N(t), (\dGamma(1-\Delta) &+ R_4) \Phi_N(t) \big\rangle \le C_\eps \int_0^t  \frac{\big\langle \Phi_N(s), (\dGamma(1-\Delta) + R_4) \Phi_N(s) \big\rangle}{(1+s)^{3/2}} \d s \nn 
\\
&+ C_\eps \Big( N^{\beta+\eps}  +N^{3\beta-1+\eps}+  \big\langle \Phi_N(t), \cN  \Phi_N(t) \big\rangle \Big).
\end{align}
Next, we estimate $\partial_t \big\langle \Phi_N(t), \cN  \Phi_N(t) \big\rangle$ by  using \eqref{eq:eq-PhiNt}, \eqref{eq:BogHam-bounds3} and the last inequality in Lemma \ref{lem:Rj}. Then we integrate the resulting bound to get
$$
\langle \Phi(t), \cN \Phi(t) \rangle \le C_\eps \int_0^t  (1+s)^{-3/2} \big\langle \Phi_N(s), (\dGamma(1-\Delta) + R_4) \Phi_N(s) \big\rangle \d s + C_\eps N^{\beta+\eps}.
$$
Substituting the latter estimate into \eqref{eq:PhiN-kin-4}, we find that
\begin{align*}
&\big\langle \Phi_N(t), (\dGamma(1-\Delta) + R_4) \Phi_N(t) \big\rangle \nn\\
&\le C_\eps \int_0^t  \frac{\big\langle \Phi_N(s), (\dGamma(1-\Delta) + R_4) \Phi_N(s) \big\rangle}{(1+s)^{3/2}} \d s + C_\eps (N^{\beta+\eps}+N^{3\beta-1+\eps})  .
\end{align*}
By a Gronwall-type inequality, we conclude that 
$$
\big\langle \Phi_N(t), (\dGamma(1-\Delta) + R_4) \Phi_N(t) \big\rangle \le C_\eps (N^{\beta+\eps}+N^{3\beta-1+\eps}).
$$ 
Since $R_4\ge 0$, the desired kinetic estimate follows. \qed
\end{proof}

\section{Norm approximation} \label{sec:sketch-proof}

\begin{proof}[of Theorem \ref{thm:main}] \smartqed  {\bf Step 1.} The desired estimate \eqref{eq:thm-mainresult} is 
$$ \|  \Psi_N(t) - U_N(t)^* \1^{\le N} \Phi(t)\|_{\gH^N}^2  \le C_\eps (1+t)^{1+\eps} N^{(2\beta+\eps-1)/2}, \quad \forall \eps>0.$$
Since $\Phi(t)=U_N(t) \Psi_N(t)$ and $U_N(t):\gH^N \to \1^{\le N}\cF_+(t)$ is a unitary operator,
$$
\|  \Psi_N(t) - U_N(t)^* \1^{\le N} \Phi(t)\|_{\gH^N} = \| U_N(t) \Psi_N(t) - \1^{\le N} \Phi(t)\| \le \| \Phi_{N}(t)- \Phi(t)\|.
$$
It remains to bound $\| \Phi_{N}(t)- \Phi(t)\|$. Using equations \eqref{eq:Bogoliubov-equation} and \eqref{eq:eq-PhiNt}, we can write \begin{align} \label{eq:final-proof-2} 
 \partial_t \| \Phi_{N}(t)-\Phi(t)\|^2 &= 2\Re \,\big \langle i \Phi_N(t),  (\widetilde H_N(t) - \bH(t))\Phi(t)\big\rangle\\
 = &  \sum_{j=0}^4 \Re \big\langle i\Phi_N(t), (R_j+R_j^*) \1^{\le N} \Phi(t) \big\rangle - 2\Re\,\big\langle i\Phi_N(t), \bH \1^{>N} \Phi(t) \big\rangle\nn
\end{align}
where $\1^{>N} := \1 - \1^{\le N}$. Next, we will estimate the right side of  \eqref{eq:final-proof-2}.
\medskip

\noindent
{\bf Step 2.} To bound the last term of \eqref{eq:final-proof-2}, we use $\Phi_N(t) \in \1^{\le N} \cF_+(t)$ to write
$$
\big\langle \Phi_N(t), \bH \1^{>N} \Phi(t) \big\rangle = \big\langle \Phi_N(t), (\bH - \dGamma(h))  \1^{>N} \Phi(t) \big\rangle.
$$
As in the proof of Lemma \ref{lem:bH-kinetic}, we can show that
\bq \label{eq:final-proof-simplebH}
\pm (\bH - \dGamma(h)) \le C( \cN + N^{3\beta}). \nn
\eq
It is a general fact that if $\pm B\leq A$ as quadratic forms, then we have  the Cauchy-Schwarz type inequality 
$|\langle f, B g\rangle|  \le 3 \langle f, A f\rangle^{1/2}\langle g, A g\rangle^{1/2}$. Consequently,
\begin{align*}
&\left| \big\langle \Phi_N(t), (\bH - \dGamma(h))  \1^{>N} \Phi(t) \big\rangle \right| \\
&\le 3   \big\langle \Phi_N(t), ( \cN + N^{3\beta})  \Phi_N(t) \big \rangle ^{1/2}  \big\langle \1^{>N}\Phi(t), ( \cN + N^{3\beta}) \1^{>N} \Phi(t) \big \rangle ^{1/2}\\
&\le 3(N + N^{3\beta})^{1/2} \big \rangle ^{1/2}  \big\langle \1^{>N}\Phi(t), (\cN+N^{3\beta})  \cN^{s}N^{-s} \1^{>N} \Phi(t) \big \rangle ^{1/2}
\end{align*}
for all $s\ge 1$. The term $\langle \Phi (t),  \cN^s\Phi(t)\rangle$ can be bounded by \eqref{eq:moment-quasi-free} and the bound on $\langle \Phi (t),  \cN \Phi(t)\rangle$ in Lemma \ref{lem:Bogoliubov-equation}. We can choose $s$ large enough (but fixed) and obtain  
\begin{align}\label{eq:final-term-one} 
\left| \big\langle \Phi_N(t), \bH \1^{>N} \Phi(t) \big\rangle \right| &\le C_{\eps}  (1+t)^\eps N^{-1}. 
\end{align}
\noindent{\bf Step 3.} To control the first term on the right side of \eqref{eq:final-proof-2}, we have to introduce a cut-off on the number of particles. Since there are at most 2 creation or annihilation operators in the expressions of $R_j$'s, we can write 
\begin{align*}
\big\langle  \Phi_N(t), (R_{j} +R_j^*)\1^{\le N} \Phi(t) \big\rangle &= \big\langle  \1^{\le M} \Phi_N(t), (R_{j}+R_j^*) \1^{\le M+2} \Phi(t) \big\rangle \\
&\quad + \big\langle  \1^{> M} \Phi_N(t), (R_{j} +R_j^*)  \1^{\le N}  \1^{> M-2}  \Phi(t) \big\rangle\nn
\end{align*}
for all $4<M<N-2$.  Then we estimate each term on the right side by Lemma \ref{lem:Rj} and the Cauchy-Schwarz type inequality as in Step 2. We obtain 
\begin{align}  \label{eq:final-proof-3}
\left| \big\langle  \Phi_N(t), (R_{j} +R_j^*)\1^{\le N} \Phi(t) \big\rangle \right| \le C(E_1 + E_2) 
\end{align}
where
\begin{align*}
E_1&= \inf_{\eta>0} \left\langle \1^{\le M} \Phi_N(t), \Big( (1+\eta) R_4 + \eta \frac{\cN^2}{N} + \frac{1+\cN}{\eta(1+t)^3} \Big) \1^{\le M} \Phi_N(t)  \right\rangle^{1/2}  \nn\\
& \qquad  \times \left\langle \1^{\le M+2} \Phi(t), \Big( (1+\eta) R_4 + \eta \frac{\cN^2}{N} + \frac{1+\cN}{\eta(1+t)^3} \Big) \1^{\le M+2} \Phi(t)  \right\rangle^{1/2}, \nn \\
E_2& = \inf_{\eta>0} \left\langle \1^{> M} \Phi_N(t), \Big( (1+\eta) R_4 + \eta \frac{\cN^2}{N} + \frac{1+\cN}{\eta(1+t)^3} \Big) \1^{> M} \Phi_N(t)  \right\rangle^{1/2} \\
&\qquad \times \left\langle \1^{> M-2} \Phi(t), \Big((1+\eta) R_4 + \eta \frac{\cN^2}{N} + \frac{1+\cN}{\eta(1+t)^3} \Big) \1^{> M-2} \Phi(t)  \right\rangle^{1/2} .
\end{align*}

To bound $E_1$, we use 
$$ \1^{\le M} R_4 \le C N^{\beta-1} \1^{\le M} \cN \dGamma(-\Delta) \le CN^{\beta-1}M \dGamma(-\Delta)$$
(see Lemma \ref{lem:Rj}) together with the kinetic estimate in Theorem \ref{lem:wHN-kinetic}, and then optimize over $\eta>0$. We get 
\begin{align*} 
E_1\le C_\eps  \Big( M N^{(2\beta+\eps-1)/2} + M^{3/2} N^{-1/2} \Big) . 
\end{align*}  
(The error term $N^{3\beta-1+\eps}$ in Theorem \ref{lem:wHN-kinetic} is 
absorbed by $N^{\beta+\eps}$ when $\beta<1/2$.)

The bound on $E_2$ is obtained using the argument in Step 2 and reads 
$$E_2 \le C_{\eps,s}N^{3\beta+1} M^{1-s/2} N^{s\eps} [\log(2+t)]^{s}.$$
In summary, from \eqref{eq:final-proof-3} it follows that
\begin{align*} 
\left| \big\langle  \Phi_N(t), (R_{j} +R_j^*)\1^{\le N} \Phi(t) \big\rangle \right| &\le C_\eps  \Big( M N^{(2\beta+\eps-1)/2} + M^{3/2} N^{-1/2} \Big) \\
& + C_{\eps,s}N^{3\beta+1} M^{1-s/2}N^{s\eps} [\log(2+t)]^{s}
\end{align*}
for all $4<M<N-2$ and $s\ge 2$. We can choose $M=N^{3\eps}$ and $s=s(\eps)$ sufficiently large (e.g. $s\ge 6 (1+\beta+\eps)/\eps$) to obtain 
\begin{align} \label{eq:final-proof-2a}
\left| \big\langle  \Phi_N(t), (R_{j} +R_j^*)\1^{\le N} \Phi(t) \big\rangle \right| \le C_\eps \Big( N^{(2\beta+9\eps-1)/2} + N^{-1}(1+t)^{\eps} \Big) .
\end{align}
\noindent 
{\bf Step 4.} Inserting \eqref{eq:final-term-one} and \eqref{eq:final-proof-2a} into \eqref{eq:final-proof-2}, we find that 
$$
\partial_t \| \Phi_{N}(t)-\Phi(t)\|^2 \le C_\eps \Big( N^{(2\beta+9\eps-1)/2} + N^{-1} (1+t)^{\eps} \Big).
$$
Integrating over $t$ and using
$$
\| \Phi_N(0)-\Phi(0)\|^2 = \langle \Phi(0), \1^{>N} \Phi(0) \rangle \le N^{-1} \langle \Phi(0), \cN \Phi(0) \rangle \le C_\eps N^{\eps-1}.
$$
we obtain
\begin{align*}
\| \Phi_{N}(t)-\Phi(t)\|^2 &\le C_\eps (1+t)^{1+\eps} N^{(2\beta+9\eps-1)/2}
\end{align*}
for all $\eps>0$. This leads to the desired estimate \eqref{eq:thm-mainresult}, as explained in Step 1.  \qed
\end{proof}

\end{document}